\documentclass[11pt]{article}
\usepackage{graphicx}

\usepackage{color}
\usepackage{natbib}
\definecolor{darkblue}{rgb}{0,0,0.8}
\usepackage{url}

{\makeatletter
 \gdef\xxxmark{%
   \expandafter\ifx\csname @mpargs\endcsname\relax % in minipage?
     \expandafter\ifx\csname @captype\endcsname\relax % in figure/caption?
       \marginpar{xxx}% not in a caption or minipage, can use marginpar
     \else
       xxx % notice trailing space
     \fi
   \else
     xxx % notice trailing space
   \fi}
 \gdef\xxx{\@ifnextchar[\xxx@lab\xxx@nolab}
 \long\gdef\xxx@lab[#1]#2{\textbf{[\xxxmark #2 ---{\sc #1}]}}
 \long\gdef\xxx@nolab#1{\textbf{[\xxxmark #1]}}
}

\usepackage{amsmath, amsthm, amssymb}

\usepackage{algorithm}
\usepackage[noend]{algpseudocode}
\def\STATE{\State}
\def\WHILE{\While}
\def\ENDWHILE{\EndWhile}
\def\COMMENT{\Comment}
\def\GETS{\gets}
 \newtheorem{theorem}{Theorem} %[section]
 \newtheorem{lemma}[theorem]{Lemma}

\newenvironment{proofof}[1]{\begin{proof}[Proof of {#1}]}{\end{proof}}
\newcommand{\opt}{\textsf{Opt}\xspace}
\newcommand{\Pure}{\textsf{Pure}\xspace}
\newcommand{\Gold}{\textsf{Gold}\xspace}
\renewcommand{\Pr}{\ensuremath{\mathbf{Pr}}}

\newcommand{\NPclass}{\textsf{NP}\xspace}
\newcommand{\DTIMEclass}{\textsf{DTIME}\xspace}

\usepackage{fullpage}
\usepackage{xspace}

\newcommand{\eps}{\ensuremath{\varepsilon}}

\newcommand{\Alg}{\ensuremath{\mathcal A}\xspace}
\newcommand{\elements}{\ensuremath{\mathcal E}\xspace}
\newcommand{\sets}{\ensuremath{\mathcal F}\xspace}
\newcommand{\solution}{\ensuremath{\mathcal S}\xspace}
\newcommand{\cov}{\ensuremath{\mathcal C}\xspace}

\renewcommand{\deg}{\operatorname{deg}}
\newcommand{\set}[1]{\ensuremath{\{#1\}}\xspace}

{\addtolength{\itemsep}{-\baselineskip}}
{\addtolength{\abovecaptionskip}{-2\baselineskip}}
{\addtolength{\belowcaptionskip}{-\baselineskip}}

\newcommand{\Ex}{\ensuremath{\mathbb E}}

\newcommand{\LJinstx}{\texttt{livej-2}}
\newcommand{\LJinst}{\texttt{livej-3}}
\newcommand{\DBLPinstx}{\texttt{dblp-2}}
\newcommand{\DBLPinst}{\texttt{dblp-3}}
\newcommand{\Guteinst}{\texttt{gutenberg}}
\newcommand{\SGuteinst}{\texttt{s-gutenberg}}
\newcommand{\Hardinst}[1]{\texttt{planted-#1}}
\newcommand{\WikiMain}{\texttt{wiki-main}}
\newcommand{\WikiTalk}{\texttt{wiki-talk}}
\newcommand{\Reuters}{\texttt{reuters}}
\newcommand{\NewsTwenty}{\texttt{news20}}

\usepackage{tabulary}
\newcolumntype{K}[1]{>{\centering\arraybackslash}p{#1}}

\newcommand{\problem}[1]{\textsf{#1}}
\newcommand{\algoname}[1]{\textsc{#1}}

\newenvironment{packed_enum}{
\begin{enumerate}
  \setlength{\itemsep}{1pt}
  \setlength{\parskip}{0pt}
  \setlength{\parsep}{0pt}
}{\end{enumerate}}

\newenvironment{packed_item}{
\begin{itemize}
  \setlength{\itemsep}{1pt}
  \setlength{\parskip}{0pt}
  \setlength{\parsep}{0pt}
}{\end{itemize}}

\newcommand{\union}{\ensuremath{\cup}}

\AtBeginDocument{%
 \abovedisplayskip=8pt plus 2pt minus 6pt
 \abovedisplayshortskip=0pt plus 3pt
 \belowdisplayskip=8pt plus 2pt minus 6pt
 \belowdisplayshortskip=7pt plus 3pt minus 4pt
}

\newcommand\myhline[1]{%
  \noalign{%
    \global\dimen1\arrayrulewidth%
    \global\arrayrulewidth#1
  }\hline
  \noalign{%
    \global\arrayrulewidth\dimen1 
  }%
}

\title{Distributed Coverage Maximization via Sketching }

\author{
	MohammadHossein Bateni\\Google Research
	\and Hossein Esfandiari\\University of Maryland
	\and Vahab Mirrokni\\Google Research
}

\date{}

\begin{document}

\sloppy
\maketitle
%\sloppy

\begin{abstract}
%Although coverage problem are an important special cases of submodular optimization problems with a wide range of applications, all the previous results for a distributed algorithm to solve this problem suffer from several shortcomings, e.g., they may ignore the complexity of communicating large subsets across machines.
In this paper, we present
distributed algorithms for coverage optimization problems
%  \problem{$k$-cover} and \problem{set cover with $\lambda$ outliers},
with almost optimal space complexity and optimal approximation
guarantees. These new algorithms also achieve an optimal communication
complexity, running in only four rounds of computation, addressing
major limitations of prior work.  While previous distributed
algorithms for submodular maximization rely on ideas of core-sets, our
algorithms are based on a new adaptive sampling and sketching technique.
We show that the proposed algorithms are implementable in various distributed
optimization frameworks such as MapReduce and RAM models.  Moreover,
%we note that 
our ideas extend to weighted variants of coverage
problems, and can solve the related \problem{dominating set} problems.

Furthermore, we perform an extensive empirical study of our algorithms
(implemented in MapReduce) on a variety of datasets.  We observe that
using sketches $30$--$600$ times smaller than the input, one can solve
the coverage maximization problem with quality very close to that of
the state-of-the-art single-machine algorithm.  Finally, we show an
application of our algorithm in large-scale feature selection.
%Coverage instances examined in this paper are an order-of-magnitude larger than the ones studied in prior papers.

\end{abstract}

\section{Introduction}
As important special cases of submodular optimization, \problem{maximum
$k$-cover} and \problem{minimum set cover} are among the most central problems in
optimization with a wide range of applications in machine learning,
document summarization, and information retrieval; e.g.,
see~\cite{CKT10,AMT,pods12,nips13}.
In order to address the need for handling large datasets, 
many techniques have been developed for distributed submodular
maximization~\cite{BST12,KMVV13,nips13,IMMM14,karbasiKDD2014,MZ15,BENW15}.
However, many of these results do not take advantage of the special structure of coverage
functions, and consequently achieve suboptimal approximation
guarantees and/or poor space complexities in terms of the size
of the (coverage) instance. In particular,
most previous results on submodular
maximization either explicitly or implicitly assume a {\em value oracle
access} to the submodular function. Such an oracle for coverage
 functions has the following form: given a subfamily of the (input)
 family, determine the size of the union of the sets in the
 subfamily. Implementing this subroutine is costly in the presence of large subsets in the
 family and/or a large ground set. Indeed communicating entire subsets across machines might be
 quite impractical. In this paper, we aim to address the above issues, and present almost optimal distributed approximation algorithms for coverage problems with optimal communication and space complexity. Before elaborating on our results, let us  describe the problem formulations, and distributed computation models discussed later.

\paragraph{Problem Formulation} Consider a ground set \elements of $m$ elements, and a family
$\sets\subseteq 2^\elements$ of $n$ subsets of the elements (i.e.,
$n=|\sets|$ and $m=|\elements|$).%
\footnote{There are two separate series of work in this area.  We use the convention of the submodular/welfare maximization  formulation~\cite{badanidiyuru2012sketching}, whereas the hypergraph-based formulation~\cite{saha2009maximum} typically uses $n, m$ in the opposite way.}  
 The {\em coverage function} $\cov$ is
defined as $\cov(\solution) = |{\cup}_{U\in\solution} U|$ for any subfamily $\solution\subseteq
\sets$ of subsets.
Given $k \geq 0$, the goal in \problem{$k$-cover} is to
pick $k$ sets from $\sets$ with the largest union size.
%We sometimes use $\opt_k$ to denote the size of the union for the optimum solution.
\problem{Set cover} asks for the minimum number of
sets from $\sets$ that together cover $\elements$ entirely.
In this paper, we also study the following variant of this
problem, called \problem{set cover with $\lambda$ outliers}%
\footnote{It is somtimes called \problem{$(1-\lambda)$-%
  partial cover} in the literature.}, where the goal is to find
the minimum number of sets covering at least a $1-\lambda$ fraction of
the elements~\elements.

\paragraph{MapReduce Model} The distributed computation model---e.g.,
MapReduce~\cite{osdi-DG04}---assumes that the data is split across
multiple machines. In each round of a distributed algorithm, the data
is processed in parallel on all machines: Each machine waits to
receive messages sent to it in the previous round, performs its
own computation, and finally sends messages to the other machines.
The total amount of data a machine processes is called its {\em load}, which had
better be sublinear in the input size.  In fact, two important factors
determine the performance of a distributed algorithm:
(i) the number of rounds of computation, and (ii) the maximum load
on any machine. These parameters have been discussed and optimized in
previous work% distributed algorithms in MapReduce framework%
~\cite{soda-KSV10,IMMM14%,ANOY14
,BBLM14,MZ15}.

\paragraph{RAM Model} Another model for handling a large amount of data is what we call the
{\em RAM model}~\cite{aho1974design}, where the algorithm has random
access to any part of the input (say, to the edge lists in the graph)
but each lookup takes constant time.  For many problems it might be
possible to judiciously and adaptively query the data, and solve the
problem.  In order to implement such a model in practice, distributed
hash-tables (such as Bigtable) have been proposed and applied in
practice~\cite{bigtable-paper}.  From a theoretical point of view, this
model is closely related to the communication complexity literature.

\paragraph{Our Contributions} In this paper, we present
distributed algorithms for \problem{$k$-cover} and \problem{set cover}
addressing several shortcomings of previously studied algorithms
and achieving optimal approximation guarantees
as well as almost optimal space and communication complexity.
To achieve this result, we present an adaptive sampling (or sketching) technique
that can be implemented in a distributed manner. We also rule out
effectiveness of various simpler sampling techniques by providing
lower bound examples.
%that achieve almost optimal approximation guarantees  with optimal space complexity without assuming value oracle access to the coverage function. % To do so, we employ a sketching technique that we developed recently~\cite{ours}.\footnote{Anonymized version provided as supplementary material.}
%\cite{owncitearxiv}.
%in an accompanied paper~\cite{ours}.  %%%% TODO: owncite
%\footnote{We have included a version without
%  author names instead of citing to a public version, to avoid
%  violating the double-blind policy.}
More precisely, our results for coverage problems are as follows:
First of all, we develop distributed algorithms
for \problem{$k$-cover} and \problem{set cover with $\lambda$
outliers}, that are almost optimal from three perspectives: (i) they
achieve optimal approximation guarantees of $1-1/e$ and $\log
{1\over \lambda}$ for the above two problems, respectively; (ii) they
have a memory complexity of $\tilde O(n)$ and also $\tilde O(n)$
communication complexity; and finally (iii) they run in a few
(constant) rounds of computation; see Table~\ref{tab:distributed} for
brief comparison of our theoretical results with prior work
(Sections~\ref{sec:MapReduce} and~\ref{sec:RAM}).  We note that the
space complexity of our algorithm is independent of the size of the
universe of elements and is only a linear function of the number of
input sets. This is crucial for tackling coverage instances with very
large sets, or large total number of elements.

\begin{table*}
  \caption{Comparison of our results to prior work.  The first three work for the more general case of \problem{submodular maximization}.}
  \label{tab:distributed}
  \begin{center}
    \begin{tabular}{ccccc} 		
      \hline
      Problem  & Credit & \# rounds & Approximation & \small Load per machine \\ 
      \hline 
%      \hline
      \problem{$k$-cover} & \cite{KMVV13} &
      $O(\frac{1}{\eps\delta}\log m)$ & $1- \frac 1 e-\eps$ &
      $O(mkn^{\delta})$   \\ 
      \problem{$k$-cover} & \cite{MZ15} & 			  $2$ & $0.54$
      & \small $\max(m k^2, mn/k)$   \\ 
      \problem{$k$-cover} & \cite{BENW15b} & 			  $1\over
      \eps$ & $1-\frac 1 e -\eps$ & $\max(m k^2,  mn/k)\over \eps$  \\ 
      \problem{$k$-cover} & Here & $4$ & $1- \frac 1 e -\eps$ & $\tilde{O}(n)$  \\  %BENW15b
      \myhline{0.1pt}
%      $\substack{\mbox{\strut Set cover}\\\mbox{\strut with outliers}}$ 
\problem{set cover with outliers}\hspace{-3mm}
& Here & $4$ & $(1+\eps)\log \frac 1 {\lambda}$ & $\tilde{O}(n)$  \\
\problem{submodular cover} & \cite{nips15} & $\Omega(n^{1\over 6})$ & $\Omega(n^{1\over 6})$ &  $\tilde{O}(mn)$\\   
\problem{submodular cover} & \cite{MZK16} & $O({\log n\log m\over \epsilon})$ & $(1+\eps)\log \frac 1 {\lambda}$ &  $\tilde{O}(mn)$\\   
      \myhline{0.1pt}
     \problem{dominating set} & Here & $4 $& $(1+\eps)\log \frac 1 {\lambda}$& $\tilde{O}(n)$    \\
      \hline	
    \end{tabular}
  \end{center}
\end{table*}

Secondly, not requiring value oracle access to the coverage function
makes our algorithms and techniques applicable to related
problems such as \problem{dominating set} with applications 
to influence maximization in social networks (see~\cite{nips13} for application). 
Indeed we give the first distributed algorithm for \problem{dominating set} that does not need
to load all edges of a node onto a single machine (Section~\ref{sec:DominatingSet}).
% All previous algorithms required to process all edges
%connected to each node of the graph on the same machine. 
This is crucial for handling graphs with nodes of very high degree.
%for %which even a single node may not fit on a machine.
Thirdly, we show that our algorithm can be implemented in both MapReduce
and RAM models, and furthermore, present extensions of our distributed 
algorithm to a number of variants of weighted coverage problems (Section~\ref{sec:weightedVariants}).

Last but not least, we demonstrate the power of our techniques via an
extensive empirical study on a variety of applications and publicly
available datasets (Section~\ref{sec:empirical}).  We observe that
sketches that are a factor $30$--$600$ smaller than the input suffice
for solving \problem{$k$-cover} with quality (almost) matching that of
the state-of-the-art single-machine algorithm; e.g., for a medium-size
dataset, we can obtain $99.6\%$ of the quality of the single-machine
stochastic-greedy algorithm using only $3\%$ of the input data.  Some
of the instances we examine in this paper are an order of magnitude
larger than the ones studied in prior work~\cite{nips13}.  Finally, we
show an application of our algorithm in large-scale feature selection
by formalizing it as a coverage problem where we aim to choose a
subset of features that {\em cover} as many {\em pairs of samples} as
possible. In doing so, we take advantage of the fact that the space
complexity of our algorithm is independent of the number of elements
in the instance, and we can solve instances of coverage problem with
very large sets.

\begin{table*}
	\caption{Comparison of results for the RAM model.}
	\label{tab:RAM}
\begin{center}
	\begin{tabular}{cccc} 		
		\hline
		Problem  & Credit & Approximation & Runtime \\ 
		\hline 
		 \problem{$k$-cover} & \cite{badanidiyuru2014fast,mirzasoleiman2014lazier} & $1- \frac 1
                 e -\eps$ & $\tilde{O}(nm)$  \\ 
		\problem{$k$-cover} & Here & $1- \frac 1  e -\eps$ & $\tilde{O}(n)$  \\ 
      \myhline{0.1pt}
		\problem{set cover with outliers} & Here & $(1+\epsilon)\log\frac1\lambda$ & $\tilde{O}(n)$  \\ 
		\hline	
	\end{tabular}
\end{center}
\end{table*}

%For medium-size datasets, we can directly compare our solution with the best  previously known offline algorithms (such as stochastic-greedy). However, for large datasets, we are unable to run offline algorithms on large datasets due to their large running time.  Notice that if we keep the whole graph as the sample the whole graph, obviously, the solution of our algorithm is equivalent to the solution of the algorithm that we run on the sketch. Thus, if at some point increasing the size of the sketch does not significantly improve the quality of the solution, the quality of the solution at that point is very close to the offline algorithm on the whole graph. 
%We run our algorithm on a set of induced subgraphs of a graph using the same parameters. Moreover, as the size of induced subgraphs grows, the performance of our algorithm increases. This means that if we find parameters $\rho$ and $\sigma$ on a subgraph of the input and apply it to the whole graph, we get at least the same performance.

\paragraph{Further Related Work}%\label{sec:intro:rel}
Although \problem{maximum $k$-cover} may be solved using a distributed
algorithm for \problem{submodular maximization}, all the prior work in
this area (have to) assume value oracle access to the submodular function, introducing
a dependence on the size of the sets in the running time of each round
of the algorithms. In this model, %for the coverage maximization
problem,
%%Chierichetti et al.~%% CITE CHANGE
\citet{CKT10} present a
$1-\frac1e$-approximation algorithm for \problem{$k$-cover} in
polylogarithmic number of rounds of computation, improvable to $O(\log
n)$ rounds~\cite{BST12,KMVV13}.  Recently randomized core-sets were
used to obtain a constant-approximation $2$-round algorithm for this
problem~\cite{MZ15,BENW15}, where the best known approximation factor
is $0.54$.  %%~\cite{MZ15}.  
In other recent work, \citet{nips15,MZK16}
give a distributed algorithm for \problem{submodular cover} (a
generalization of \problem{set cover}) in the MapReduce framework, however,
their algorithm runs in superconstant number of rounds. and compared
to the result presented here, they have much larger
space complexity when it's applied to \problem{set cover}.

\paragraph{More Notation}
Coverage problems may also be described via a bipartite graph $G$,
with the two sides corresponding to \sets and
\elements, respectively. 
The edges of $G$ correspond to pairs $(S, i)$ where $i \in S \in \sets$.
%For each set $S\in \sets$,
%there are $\vert S\vert$ edges in $G$ from the vertex corresponding to
%$S$ to vertices corresponding to elements $i \in S$.  
For simplicity, we assume that there is no isolated vertex in
\elements.
As is customary, we let $\Gamma(G, V')$ denote the set of neighbors of
vertices $V'$ in $G$.  When applied to a bipartite graph $G$ modeling
a coverage instance, we can write the coverage problem as
$\cov(\solution) = |\Gamma(G, \solution)|$ for any
$\solution \subseteq \sets$.

\section{Distributed Algorithms} \label{sec:MapReduce}
In this section we present distributed algorithms for \problem{$k$-cover} and
\problem{set cover with $\lambda$ outliers}. 
We aim to develop algorithms that only need $\tilde{O}(n)$ space per
machine.  As a first attempt, if we want to apply the distributed
submodular optimization results to our problems
(e.g., \algoname{DistGreedy}~\cite{nips13} or composable core-set
algorithm~\cite{MZ15,BENW15}), the underlying algorithms would
distribute {\em subsets} across machines. The main issue with such an
approach is that sending whole subsets does not scale well for
large subsets.  A natural way to deal with the issue of large subsets
is to subsample elements while sending those sets around, and a
natural sampling technique would be {\em uniform sampling}.  We
first rule out applicability of such simple sampling schemes for this
problem. In particular, we present a hardness example for which the size
(i.e., the number of edges) of the instance on each machine has to be
$\Omega(nk)$ to obtain a bounded approximation guarantee.

\begin{theorem}
Pick arbitrary numbers $n, \beta \geq 1$ and $k\leq n/2$. Let $\Alg$
be an algorithm that samples elements uniformly at random and reports
an arbitrary optimum solution to \problem{$k$-cover} on the sampled
instance. If the number of edges sampled by $\Alg$ does not exceed
$nk/\beta^2$, its approximation factor is at most $\frac
{2}{\beta+1}$.
\end{theorem}
\begin{proof}
  Consider the following example with $k$ {\em bonus sets} and $n-k$
  normal sets. Moreover, we have $\beta n$ {\em special elements} and
  $n$ normal ones. Each set has edges to all normal elements, and each
  bonus set has edges to $\beta n/k$ unique bonus elements. Notice
  that the optimum \problem{$k$-cover} solution picks all the $k$
  bonus sets, and covers all the $(\beta+1)n$ elements.
	
  Note that each normal element has $n$ edges.  Since \Alg\ samples at
  most $nk/\beta^2$ edges, no more than $k/\beta^2$ normal elements in
  expectation make it to the sample. In other words, each element is
  sampled with probability at most $\frac{k}{\beta^2n}$. Therefore,
  $\Alg$ samples at most $\frac{k}{\beta^2n}\times \beta n = n/\beta$
  bonus elements, in expectation.
	
  Indeed, if $\Alg$ do not pick any bonus elements corresponding to a
  bonus set $S$, in the sampled graph the set $S$ covers the same
  elements as any normal set does.  Thus $\Alg$ might pick a normal
  set instead of $S$ in an arbitrary optimum solution on the sampled
  graph. Notice that $\Alg$ samples at most $n/\beta$ bonus elements
  in expectation, which corresponds to no more than $n/\beta$ distinct
  bonus sets, in expectation. Hence there is an optimum solution on
  the sampled graph with $n/\beta$ bonus sets and $n-n/\beta$ normal
  sets in expectation. The expected total number of elements in this
  solution is $n + \frac{n}{\beta}\times \beta n =2n$.
\end{proof}

This observation suggests that any distributed algorithm should employ
a more nuanced sampling (sketching) technique. To this end, we invoke
a recent technique of ours~\cite{ours}\footnote{In a recent
work~\cite{ours}, we study streaming algorithms for coverage functions
and in particular show that an $\alpha$ approximate solution to
\problem{$k$-cover} on a sketch with the above properties is an
$\alpha-\epsilon$ approximate solution on the actual input, with
probability $1-e^{-\delta}$. This recent work is included as
supplementary material and will be made available online.}:
if we can develop a distributed algorithm with $\tilde{O}(n)$
space that outputs a sketch (denoted by $H_{\leq n}(k,\eps,\delta'')$,
or simply $H_{\leq n}$) satisfying {\em three special properties}, we
can prove tight approximation guarantees for the following algorithm:
solve the problem by running a greedy algorithm on the sketch.

\begin{algorithm}%[!h]
  \textbf{Input:} Input graph $G$ and parameters $k$, $\eps\in (0,1]$, $\delta''$.\\
  \textbf{Output:} Solution to the coverage problem.
\def\STATE{\Statex\hspace{-3.8ex}}
  \begin{algorithmic}[1]
    \STATE Let $h:\elements\mapsto[0, 1]$ be a uniform, independent hash function.

    \STATE\textbf{Round 1:} Send the edges of each element to a
    distinct machine. Let $\tilde{n} = \frac{24n\delta\log(1/\eps)\log
      n}{(1-\eps)\eps^3}$. For each element $v$, if $h(v)\leq \frac
    {2\tilde{n}} m$, the machine corresponding
    to $v$ sends $h(v)$ and its degree to machine one;
    it does nothing otherwise.
    \STATE\textbf{Round 2:} Machine one iteratively selects
    elements with the smallest $h$ until the sum of the degrees of
    the selected vertices reaches $\tilde n$. Then it
    informs the machines corresponding to selected elements.
    \STATE\textbf{Round 3:} For each selected element $v$, if the
    degree of $v$ is less than $\Delta = \frac{n\log(1/\eps)}{\eps k}$, machine
    $v$ sends all its edges to machine one. Otherwise, it 
    sends $\Delta$ arbitrary edges to machine one.
    \STATE\textbf{Round 4:} Machine one receives the
    sketch  $H_{\leq n}$ and solves the coverage problem on it by applying a greedy algorithm. %(k,\eps,\delta'')$.
  \end{algorithmic}
  
  \caption{Distributed algorithm for \problem{$k$-cover}} %for $H_{\leq n}(k,\eps,\delta'')$}%{Algorithm \ref{PF:Alg1}}
  \label{Alg:MR}
\end{algorithm}

Here we develop Algorithm~\ref{Alg:MR}, a four-round distributed
algorithm\footnote{Number of rounds were not optimized due to
readability.}, and prove the main result of this section, by showing
that the output of this algorithm satisfies those three properties
with high probability in a distributed setting using only
$\tilde{O}(n)$ space. More formally, we prove the following.

\begin{theorem}\label{thm:mr:kcover} 
  With probability $1-\frac 2 n$, Algorithm~\ref{Alg:MR} outputs a
  $(1-\frac 1 e -\epsilon)$-approximate solution
  to \problem{$k$-cover}, and no machine uses more than $\tilde{O}(n)$
  space in this algorithm.
\end{theorem} 

%While our \problem{$k$-cover} algorithm runs a greedy algorithm on the sketch, our algorithm for \problem{set-cover with $\lambda$ outliers} makes logarithmically many guesses on the number of sets in the solution, constructs $H_{\leq n}$ sketches for each (simultaneously), and solves the problem on each resulting sketch. 

The proof has two ingredients.  First of all, we show that the sketch
$H_{\leq n}$ computed in this algorithm satisfies the following three
properties: Given parameters, $\epsilon$ and $\delta''$, (1) elements
are sampled uniformly at random, (2) the degree of each element is
upper bounded by $\frac{n\log(1/\eps)}{\eps k}$, and (3) the total
number of edges is at least $\frac{24 n\delta\log(1/\eps)\log
n}{(1-\eps)\eps^3}$, where $\delta = \delta'' \log \log_{1-\eps}
m$.
%We denote this sketch by $H_{\leq n}$\footnote{When clear from  context, we drop parameters $k,\eps$, and $\delta''$.}.
Secondly, we need to show that the algorithm  uses $\tilde{O}(n)$ space per machine.
%We start by proving the first claim in the following lemma:
The following lemma summarizes properties of the algorithm that pave
the way for the proof of the theorem.

%Then, we apply Theorem~\ref{thm:sketchAlg} to solve \problem{$k$-cover}
%and \problem{set cover with $\lambda$ outliers} in another round. 

\begin{lemma}\label{lm:cher}
  Given are a graph $G(\sets\union\elements,E)$ along with $k$,
  $\eps\in (0,1]$ and $\delta''\in (0,1]$.  Then with probability
  $1-1/n^2$, $ H_{\leq n}$ 
\vspace{-2ex}
  \begin{packed_item}  %% TODO: owncite
  \item 
    %(1)
    no element with hash value exceeding
%    $H_{\leq n}(k,\eps,\delta'')$ 
    $\frac {2\tilde{n}} m$,
    and 
  \item 
    %(2)
     at most $3\tilde{n}$ edges with hash value
    exactly $\frac {2\tilde{n}} m$,
  \end{packed_item}\vspace{-2ex} 
  where $\tilde{n} = \frac{24n\delta\log(1/\eps)\log
    n}{(1-\eps)\eps^3}$.
\end{lemma}
\vspace{-4ex}

\begin{proof}
%  Note that Algorithm~\ref{Alg:Hn} stops adding elements to $H$ as
%    soon as the number of edges in the sketch hits $\tilde{n}$. Thus,
%    excluding its last element, $H$ has less than $\tilde{n}$
%    edges, hence less than $\tilde{n}$ elements. Therefore, the sketch
%    $H$ contains at most $\tilde{n}$ elements.
    
    Note that $H_{\leq n}$ requires to have only $\tilde{n}$ edges. Clearly to satisfy this
    it is sufficient to have $\tilde{n}$ elements.
  In the rest of the proof we show that, with probability $1-1/n$, the
  number of elements with hash value less than $\frac {2\tilde{n}} m$
  is within the range $[\tilde{n}, 3\tilde{n}]$.  The lower bound
  together with the fact that $H$ contains at most $\tilde{n}$
  elements gives us the first part of the theorem. The upper bound
  directly proves the second part of the theorem.

  For every element $v\in\elements$, let $X_v$ be the binary random variable
  indicating whether $h(v) < \frac {2\tilde{n}} m$,
  and let $X=\sum_{v \in \elements} X_v$ denote the number of
  elements with hash value less than $\frac {2\tilde{n}} m$. The
  Chernoff bound gives
  \vspace{-2ex}
  \begin{align}\label{eq:cher0}
    \Pr\Big( |X-\Ex[X]|\geq \frac 1 2
    \Ex[X] \Big) &\leq 
    2 \exp\Bigg(-\frac{\frac 1 4 \Ex[X]}{3}\Bigg)\nonumber\\&=2
    \exp\Big(-\frac{\Ex[X]}{12}\Big).
  \end{align}
  \vspace{-3ex}

  Remark that $h$ is a uniform mapping to $[0,1]$.
  Thus, $\Pr[h(v)\leq \frac {2\tilde{n}} m] = \frac {2\tilde{n}} m$ for any
  element $v$, so
  we have
  \vspace{-1ex}
  \begin{align}\label{eq:E[X]}
    \Ex[X] = \sum_{v \in \elements} \Ex[X_v] 
    = \sum_{v \in \elements} \frac {2\tilde{n}} m = 2\tilde{n}.
  \end{align}
  Putting \eqref{eq:cher0} and~\eqref{eq:E[X]} together gives
  us
  \begin{align*}
    \Pr\Big( |X-2\tilde{n}|\geq \tilde{n} \Big) &
    \leq 2 \exp\Big(-\frac{\tilde{n}}{6}\Big) \\
    &
= 2\exp\Big(-\frac{4 n\delta\log(1/\eps)\log n}{(1-\eps)\eps^3}\Big) \\
    &
\leq 2\exp\Big(-2\log n-1 \Big)\\
    &
<\exp\Big(-2\log n \Big)= \frac 1 {n^2}. 
  \end{align*} 
  Thus with probability $1-\frac 1 {n^2}$ we have
  $\tilde{n} \leq X\leq 3\tilde{n}$. %, which completes the proof.
 \end{proof}

%{\bf \noindent Proof of {Theorem~\ref{thm:mr:kcover}}.}
\begin{proofof}{Theorem~\ref{thm:mr:kcover}}
We first show that Algorithm~\ref{Alg:MR} uses $\tilde{O}(n)$ space per machine.
%,  with probability \xxx{y}. 
  Next we prove that the algorithm constructs by Round~4 a sketch
  satisfying the desirable three properties mentioned above.  As a
  result, Lemma~\ref{thm:sketchAlg} guarantees that invoking the
  greedy algorithm in Round~4 produces the promised solution.
	
  The degree of each element is at most $n$, the number of sets; thus,
  the space consumption of each machine in the first and third rounds
  is $\tilde{O}(n)$. In the second round
  machine number 1 receives $\tilde{O}(1)$ bits from each machine
  independently with probability $\frac{2\tilde{n}}m$. Using the
  second condition of Lemma~\ref{lm:cher}, the number of messages that
  this machine receives is at most $3\tilde{n}$. Therefore, this
  machine uses $\tilde{O}(n)$ space. The number of edges machine one
  receives in the fourth round is at most $\tilde{n}+n = 
  \tilde{O}(n)$.

  By the first condition of Lemma~\ref{lm:cher}, no element in
  $H_{\leq n}$ has hash value more than $\frac{2\tilde{n}}{m}$. Thus
  the machines with no output in the first round do not miss any
  elements of $H_{\leq n}$. Then the set of elements selected by
  machine one in round two is the same as in $H_{\leq n}$. Therefore,
  what machine one receives in the fourth round is $H_{\leq n}$.
  Discussion at the beginning of the proof finishes the argument.
\end{proofof}

%\xxx{We need a better organization in this section.  Good connectors
%  and outlines, etc.  Still we need to make things more concise.}

While our algorithm for \problem{$k$-cover} runs a greedy algorithm on
the sketch, our algorithm for \problem{set-cover with $\lambda$
outliers} makes logarithmically many guesses on the number of sets in
the solution, constructs $H_{\leq n}$ sketches for each
(simultaneously), and solves the problem on each resulting sketch. 
The proof of the following theorem is deferred to the full version.

\begin{theorem}\label{thm:mr:epscover}
  There exists a four-round distributed algorithm that reports a
  $(1+\eps)\log \frac 1 {\lambda}$-approximate solution
  to \problem{set cover with $\lambda$ outliers}, with probability
  $1-\frac 2 n$. Moreover, each machine uses $\tilde{O}(n)$ space in
  the algorithm.
\end{theorem}

\section{Algorithms for RAM Model} \label{sec:RAM}
In this section we explain how our results apply to the RAM model, as
well.  Recall that in this model, we have random access to the edge
lists, however, each access takes $O(1)$ time.

\begin{algorithm}%[!h]
	\textbf{Input:} Graph $G(\sets\union\elements,E)$ and numbers $k, \eps\in (0,1], \delta''$.\\
	%$k$, $\eps\in (0,1]$, $\delta''$.\\
	\textbf{Output:} Sketch $H(V_H, E_H) = H_{\leq n}(k,\eps,\delta'')$.
	
	\begin{algorithmic}[1]
		\STATE $\delta \GETS \delta'' \log \log_{1-\eps} m$
		%
		%    \STATE Pick uniform, independent hash function $h:\elements\mapsto[0, 1]$
		\STATE  $h:\elements\mapsto[0, 1]$ uniform, independent hash function
		\label{Alg:Hn:l2}
		\STATE $V_H \GETS \sets$ and $E_H\GETS\emptyset$  \COMMENT{Initialize}
		\WHILE {$|E_H| < \frac{24 n\delta\log(1/\eps)\log n}{(1-\eps)\eps^3}$}
		\STATE $v\GETS\arg\min_{v\in\elements\setminus V_H} h(v)$%
		\label{Alg:Hn:l5}
		\STATE $V_H \GETS V_H\union\{v\}$
		\STATE Add $\min(\frac{n\log(1/\eps)}{\eps k}, |\Gamma_G(v)|)$ edges of $v$ to $E_H$
		\label{Alg:Hn:l7}
		\ENDWHILE
		%$H_{\leq n}(k,\eps,\delta'')$
		%
	\end{algorithmic}	
	\caption{Abstract construction of the sketch} %$H_{\leq n}(k,\eps,\delta'')$}%{Algorithm \ref{PF:Alg1}}
	\label{Alg:Hn}
\end{algorithm}

\begin{theorem}
  There exists an algorithm that, given random access to the edge
  lists of coverage instance $G(\sets\union\elements, E)$, computes
  the sketch $H = H_{\leq n}$ in time $\tilde O(n)$.
\end{theorem}

\begin{proof}
  We show how Algorithm~\ref{Alg:Hn} can run in the RAM
  model. Since $|E_H| = \tilde O(n)$ at the end, total work done in
  Line~\ref{Alg:Hn:l7} is $\tilde O(n)$.  In Line~\ref{Alg:Hn:l2} we
  do not need to define the hash function explicitly.  When
  Line~\ref{Alg:Hn:l5} seeks the next vertex, it is equivalent to pick
  a random new vertex.  We only need to keep a list of already
  selected vertices to avoid repetition.
\end{proof}

Once the sketch is constructed we can run a sequential algorithm on
the sketch (or sketches) to solve \problem{$k$-cover} %, \problem{set cover} 
and \problem{set cover with outliers}.  The proof is almost
identical to those of Theorems~\ref{thm:mr:kcover}
and~\ref{thm:mr:epscover} and is omitted.

\begin{theorem}
  There is an $\tilde O(n)$-time, $1-\frac1e-\epsilon$-approximation
  algorithm in the RAM model for  \problem{$k$-cover}.
\end{theorem}

\begin{theorem}
  There is an $\tilde O(n)$-time algorithm in the RAM model that finds
  a $(1+\epsilon)\log\frac1\lambda$-approximate solution
  for \problem{set cover with $\lambda$ outliers}.
\end{theorem}

\section{Dominating Set} \label{sec:DominatingSet}
In \problem{dominating set} problem, we are given a graph $G(V,E)$ and we aim to
find the minimum number of vertices $S\subseteq V$ such that 
$S \union \Gamma_G(S) = V$.
%every other vertex is the neighbor of one selected. 
We say a set of vertices $S \subseteq V$ is a {\em dominating set with $\lambda$ outliers} if 
$|S \union \Gamma_G(S)| \geq (1-\lambda)|V|$.
%at most a $\lambda$ fraction of vertices of $G$ neither exists nor has a neighbor in $S$. 
A set of vertices $S$ is an $\alpha$-approximate
solution to \problem{dominating set with $\lambda$ outliers} if (1) it is a
dominating set with $\lambda$ outliers, and (2) $|S|$ is
at most $\alpha$ times the size of the smallest dominating set.

The following theorem provides the first distributed algorithm for
\problem{dominating set with $\lambda$ outliers}.

%\xxx{our dominating set instances are based on a three-hop expansion.
%the original problem is not that interesting from an experimental
%point of view because not only do the two sides have the same size,
%but also there are not that many edges.}

\begin{theorem} \label{thm:Ldominatin}
  There exists a four-round distributed algorithm that reports a
  $(1+\eps)\log \frac 1 {\lambda}$-approximate solution to \problem{dominating
  set with $\lambda$ outliers}, with probability $1-\frac 2 n$, while
  each machine use only $\tilde{O}(n)$ space.
\end{theorem}
\begin{proof}
  We give a reduction from
  \problem{dominating set with $\lambda$ outliers}
  to
  \problem{set cover with $\lambda$ outliers}.
  Let graph $H(V, E)$ be an instance
  of the former. We construct an instance $G(\sets\union\elements,E')$
  of the latter problem.  The $n$
  sets in $\sets$ correspond to the vertices of $H$. Similarly,
  the elements in $\elements$ correspond to the vertices in $H$.
  An edge  $(a, b): a \in \sets, b\in \elements$ appears in $G$
   if $a=b$ or $(a, b) \in E$.

%  For any vertex $v$ of $H$, the union of $v$ and its neighbors in $H$ forms
%  the set of elements that the set $v$ covers in $G$.
%  Similarly the elements that a subset
%  $S\subseteq \sets$ covers is exactly the set of vertices in $S$ and
%  all their neighbors. Thus any set cover with $\lambda$
%  outliers on $G$ is a dominating set with $\lambda$ outliers on $H$,
%  and vice versa.
  Any solution $\solution\subseteq\sets$ to the \problem{set-cover}
  instance $G$ corresponds to a subset of $V$ that dominate
  the vertices the corresponding sets cover, hence a one-to-one
  correspondence between the solutions of the two problems.
\end{proof}

The input for \problem{$k$-dominating set} includes a number $k$ in
addition to the graph $G$.  The goal is to select $k$ vertices that
maximize the number of {\em dominated} vertices. Then a subset of
vertices $S$ is an $\alpha$-approximate solution
to \problem{$k$-dominating set} if it covers $\alpha$ times that of
the optimum. % solution.

The following theorem provides the first distributed algorithm for
\problem{$k$-dominating set}.

\begin{theorem} 
	There exists a four-round distributed algorithm that reports a
	$1-\frac1e - \eps$-approximate solution to \problem{$k$-dominating
	set}, with probability $1-\frac 2 n$.
	Moreover, each machine uses only $\tilde{O}(n)$ space in the algorithm.
\end{theorem}
\begin{proof}
	Similarly to the proof of Theorem~\ref{thm:Ldominatin}, we
	give a reduction from \problem{$k$-dominating set}
	to \problem{$k$-cover}. From an instance $H$ of the former
	problem, we construct an instance $G(\sets\union\elements,E')$ of the latter.
	Once again, each set in $\sets$ corresponds to a
	vertex in $H$, as is each element in $\elements$. We place an
	edge between $a \in \sets$ and $b\in \elements$ if and only if
	$a=b$ or there is an edge between $a$ and $b$ in $H$.
	
	For any set $v\in\sets$, the set of elements covered by $v$ is
	exactly the union of $v$ and all its neighbors in
	$H$. Similarly, for any subset $S\subseteq \sets$, the
	elements that $S$ covers is the union of vertices in $S$ and
	all their neighbors. Therefore any \problem{$k$-cover}
	solution on $G$ corresopnds to a \problem{$k$-dominating set}
	solution with the same coverage, and vice versa.
\end{proof}

\section{Weighted Variants} \label{sec:weightedVariants}
In this section we extend our results to three variants of the
coverage problem.  In \problem{element-weighted $k$-cover}, a weight
$w_v$ is associated with each element $v \in \elements$, and the
objective is to maximize the total weight of covered elements.  
A \problem{fractional $k$-cover} instance has quantity $\alpha_{u,v}
\in [0,1]$ for each $S\in\sets, v\in\elements$, denoting that set $S$
covers $\alpha_{S,v}$ fraction of element $v$.  A solution
$\solution\subseteq \sets$ covers $\max_{S\in\solution} \alpha_{S,v}$
fraction of element $v$. Here the objective is to find a solution
$\solution\subseteq \sets$ of size $k$ that maximizes $\sum_{v \in
  \elements} \max_{S\in \solution} \alpha_{S,v}$.
Finally in \problem{probabilistic $k$-cover}, quantity $\alpha_{S,v}
\in [0,1]$ is provided for each pair of $S \in \sets$ and $v \in
\elements$: set $S$ covers element $v$ with probability
$\alpha_{S,v}$. A solution $\solution\subseteq \sets$ covers $1-
\prod_{S\in \solution}(1- \alpha_{S,v})$ fraction of element $v$. The
objective then is to find a solution $\solution\subseteq \sets$ of
cardinality $k$ that maximizes $\sum_{v \in \elements} \big(1-
\prod_{S\in \solution}(1- \alpha_{S,v})\big)$.

In the first problem, for simplicity we assume that all weights are
integers upper-bounded by a number $U$. Similarly, in the second and
the third problems, we assume that $\alpha_{S,v}$ is a factor of $1/U$
for any $v\in S\in\sets$.

\begin{theorem}\label{thm:mr:W} 
	There exists a four-round distributed algorithm that finds a
        $(1-\frac 1 e -\epsilon)$-approximate solution to
        \problem{element-weighted $k$-cover}, with probability
        $1-\frac 2 n$. Moreover, each machine uses
        $\tilde{O}(n)$ space in this algorithm.
\end{theorem} 
\begin{proof}
Let's replace an element $v \in \elements$ of weight $w_v$ with $w_v$
copies of $v$ of weight one.  This does not change the coverage of
any solution, however, it may significantly increase the size of the
problem---we can end up with $U m$ elements.  The sketch size, though,
is logarithmic in terms of the number of elements, though, and we can also
sample the new elements implicitly: simply pick each element with
probability proportional to its weight.
\end{proof}

\begin{theorem}\label{thm:mr:F} 
	There exists a four-round distributed algorithm that reports a
        $(1-\frac 1 e -\epsilon)$-approximate solution to
        \problem{fractional $k$-cover}, with probability $1-\frac 2
        n$. Moreover, no machine uses more than $\tilde{O}(n)$ space
        in this algorithm.
\end{theorem} 
\begin{proof}
 Once again we reduce to unweighted \problem{$k$-cover} and show how
 to perform the sampling implicitly.  Replace each element $v$ with
 $U$ copies and connect the first $\alpha_{S,v}U$ copies of $v$ to
 $S$.  We observe that the coverage of any solution grows by exactly a
 factor $U$.

 The sketch size is logarithmic in terms of the number of elements,
 which grows by a factor $U$. To sample an element form the unweighted
 \problem{$k$-cover} instance uniformly at random, we equivalently
 first sample an element from the original \problem{fractional
   $k$-cover} instance uniformly at random, and then pick an index
 from $[1, U]$ to decide how many of its copies should appear in the
 sketch.
\end{proof}

\begin{theorem}\label{thm:mr:P} 
	There exists a four-round distributed algorithm, using $\tilde
        O(n)$ space per machine, that finds a $(1-\frac 1 e
        -2\epsilon)$-approximate solution to \problem{probabilistic $k$-cover}
        with probability $1-\frac 3 n$.
\end{theorem} 
\begin{proof}
  Similarlay we transform an instance of \problem{probabilistic
    $k$-cover} to one of \problem{$k$-cover}: substitute each element
  $v \in \elements$ with $\zeta = \frac{12 (n+1+\log n)U}
  {\epsilon^2}$ copies of $v$, and for each set $S$ that contains $v$
  we connect $S$ to each copy of $v$ with probability $\alpha_{S,v}$.

  We show that with probability $1-\frac1n$, for all solutions
  $\solution\subseteq \sets$, the coverage of $\solution$ in the
  \problem{$k$-cover} instance is within a factor $\zeta(1\pm\epsilon
  / 2)$ of that in the original \problem{probabilistic $k$-cover}
  instance.  Fix a solution $\solution$, and let $\beta_v = 1-
  \prod_{S\in\solution}(1- \alpha_{S,v})$. The Chernoff bounds gives
  for $X$, the number of copies of $v$ covered by $\solution$, as
  follows.
  \begin{align*}
    &\Pr\Big( |X-\zeta \beta_v |\geq \zeta \beta_v \epsilon/2 \Big) 
     \quad\leq 2 \exp\Big(-\frac{ \zeta  \beta_v \epsilon^2  }{12}\Big) \\
    &= 2 \exp\Big(-\frac{ \frac{12 (n+1+\log n)U} {\epsilon^2}  \beta_v \epsilon^2  }{12}\Big) \\
    &\leq 2 \exp\Big(-\frac{ \frac{12 (n+1+\log n)} {\epsilon^2} \epsilon^2  }{12}\Big) \qquad \text{since $\beta_v \geq 1/U$,}  \\
    & \leq 2 \exp\Big(-{ { n+1+\log n}  }\Big) %\\
    \quad \leq 2^{-n}/n.
  \end{align*}
	
  There are $2^n$ choices for $\solution$, hence for all solutions
  $\solution\subseteq \sets$, the coverage of $\solution$ on the
  \problem{$k$-cover} instance is within the promised interval with
  probability $1-\frac1n$.

 Since the number of elements in the \problem{$k$-cover} instance is
 at most $\zeta m$, the size of the sketch grows 
 logarithmically in $U$. To sample an element form the
 \problem{$k$-cover} instance uniformly at random, we  sample an
 element from the original \problem{fractional $k$-cover} instance
 uniformly at random and connect it to each set $S\in \sets$ with
 probability $\alpha_{S,v}$.
\end{proof}

\section{Empirical Study and Results} \label{sec:empirical}
We begin this section by a brief overview of the datasets and corresponding applications 
used in our empirical study (see Table~\ref{tab:datasets}), and then move
to the methodology as well as the experiment results.  Detailed
information on the datasets is given in Appendix~\ref{sec:full-data}.

\begin{table}[ht]
\caption{General information about our datasets.}\label{tab:datasets}
\vskip2mm
\small
\centering
\begin{tabular}{l>{\hskip-5mm}c>{\hskip-1mm}r>{\hskip-1mm}rr}
\hline
Name    & Type           & \multicolumn{1}{c}{$|\sets|$} & \multicolumn{1}{c}{$|\elements|$} & \multicolumn{1}{c}{$|E|$} \\
\hline
\LJinst       & dominating set &
4M   &  4M   & 73B    \\
\LJinstx       & dominating set &
4M   &  4M   & 3.4B    \\
\DBLPinst     & dominating set &     
320K     &    320K   &  330M    \\
\DBLPinstx     & dominating set &     
320K     &    320K   & 27M    \\
\Guteinst     & bag of words &   42K   &  100M   &  1B    \\
\SGuteinst    & bag of words &   925   &  11M   &  27M \\
\Reuters      & bag of words & 200K & 140K & 15M \\
\Hardinst{A}  & planted coverage &   10K   &    10K   &  1.2M    \\
\Hardinst{B}  & planted coverage&   100K  &    1M   &  1.2B    \\
\Hardinst{C}  & planted coverage &   100K  &    10M  &  2.4B    \\
\Hardinst{D}  & planted coverage &   101K &    10M &  1.2B \\
\WikiMain   & contribution graph & 2.9M & 11M & 75M \\
\WikiTalk   & contribution graph & 1.7M & 1M & 7.3M \\
\NewsTwenty & feature selection & 1.4M & 200M & 4.3B \\  %% pair graph
%\NewsTwenty & feature selection & 1.4M & 20K & 9M \\  %% original graph
%% \LJinst       & dominating set &
%% 3,997,962   &  3,997,962   & 72,803,204,325    \\
%% \LJinstx       & dominating set &
%% 3,997,962   &  3,997,962   & 3,377,182,611    \\
%% \DBLPinst     & dominating set &     
%% 317,080     &    317,080   &  333,505,724    \\
%% \DBLPinstx     & dominating set &     
%% 317,080     &    317,080   & 27,437,914    \\
%% \Guteinst     & bag of words &   41,716   &  99,949,091   &  1,068,977,156    \\
%% \SGuteinst    & bag of words &   925   &  10,620,424   &  27,337,479 \\
%% \Reuters      & bag of words & 199,328 & 138,922 & 15,334,605 \\
%% \Hardinst{A}  & planted &   10,100   &    10,000   &  1,220,000    \\
%% \Hardinst{B}  & planted &   100,100  &    1,000,000   &  1,201,100,000    \\
%% \Hardinst{C}  & planted &   100,500  &    10,000,000  &  2,410,100,000    \\
%% \Hardinst{D}  & planted &   101,000  &    10,000,000 &  1,210,100,000 \\
%% \WikiMain   & contribution graph & 2,953,425 & 10,619,081 & 75,151,304 \\
%% \WikiTalk   & contribution graph & 1,736,343 & 1,017,617 & 7,299,920 \\
\hline
\end{tabular}
\end{table}

Our empirical study is based on five types of instances covering a variety of applications:
\begin{packed_enum}
\item \problem{Dominating-set} instances are formed by considering vertices
of a graph as {\em sets} and their two- or three-hop neighborhoods as
{\em elements} they dominate. The  \problem{Dominating-set} problem is
motivated by sensor placement and influence maximization applications~\cite{nips13}. 
\item The ``bag of words'' instances correspond to documents and
  bigrams they contain.  The goal is pick a few documents that cover
  many bigrams together. This instance highlights the application of coverage maximization
  in document summarization, or finding representative entities in a corpus~\cite{nips13}.
\item We have synthetic ``planted set cover'' instances that are synthetically generated,
and known to be hard for greedy algorithms.
\item ``Contribution graphs'' model interaction between users on a set
  of documents.  We seek a small subset of users that collectively
  have contributed to a majority of documents. This, in turn, has application in 
  team formation~\cite{dblp-snap}.
\item A \problem{feature-selection} instance proposes a
  \problem{column subset selection} problem on a matrix of news
  articles and their features. This application is described in Section ~\ref{sec:featureselection}.
\end{packed_enum}

We remark that, to the best of our knowledge, some of these datasets
are an order of magnitude larger than what has been considered in prior
work.

\subsection{Approach}

\newcommand{\stochgreedy}{\textsc{StochasticGreedy}}

Recall that the sketch construction is based on two types of prunings
for edges and vertices of the input graph:
\vspace{-2ex}
\begin{packed_item}
	\item 
          %(1)
          subsampling the elements, and
	\item 
          %(2)
          removing edges from large-degree elements.
\end{packed_item}
\vspace{-1ex}
The theoretical definition of the sketch provides (i)~the
probability of sampling an element, and (ii)~the upper bound on the
degree of the elements. Though these two parameters are almost
tight in theory, in practice one can use smaller values to get 
desirable solutions.
Here we parameterize our algorithm by $\rho$ and $\sigma$,
where $\rho$ is the probability of sampling elements, and
$\sigma$ is the upper bound on element degrees. We
investigate this in our experiments.
%We run our
%algorithm on different datasets to examine its speed and efficiency.

The \stochgreedy\ algorithm~\cite{mirzasoleiman2014lazier} achieves
$1-\frac 1 e -\eps$ approximation to maximizing monotone submodular
functions (hence coverage functions) with $O(n \log(1/\eps))$
calls to the submodular function. This is theoretically the
fastest known $1-\frac 1 e -\epsilon$ approximation algorithm for
coverage maximization, and is also the most efficient in
practice for miximizing monotone submodular functions, when the input
is very large. We plug it into our MapReduce algorithm,
which then runs much faster, while losing
very little in terms of quality. % compared to \stochgreedy.
For smaller instances we compare our algorithm to \stochgreedy, but
for larger ones we provide convergence numbers to argue that the two
should get very similar coverage results.

\paragraph{LiveJournal social network} We try different values for
$\rho, \sigma, k$ when running our algorithm on \LJinst; see %the
% results in 
Figure~\ref{fig:lj}. %Notice that
For small $k$, the result improves as $\sigma$ grows, but
increasing $\rho$ has no significant effect. % on the outcome.
On the other hand, the improvement for larger $k$ comes from
increasing $\rho$ while $\sigma$ is not as important. This observation
matches the definition of our sketch, in which the degree bound is
decreasing in $k$ and the sampling rate is increasing in $k$.

\begin{figure*}
	\centering
	\includegraphics[scale=0.48]{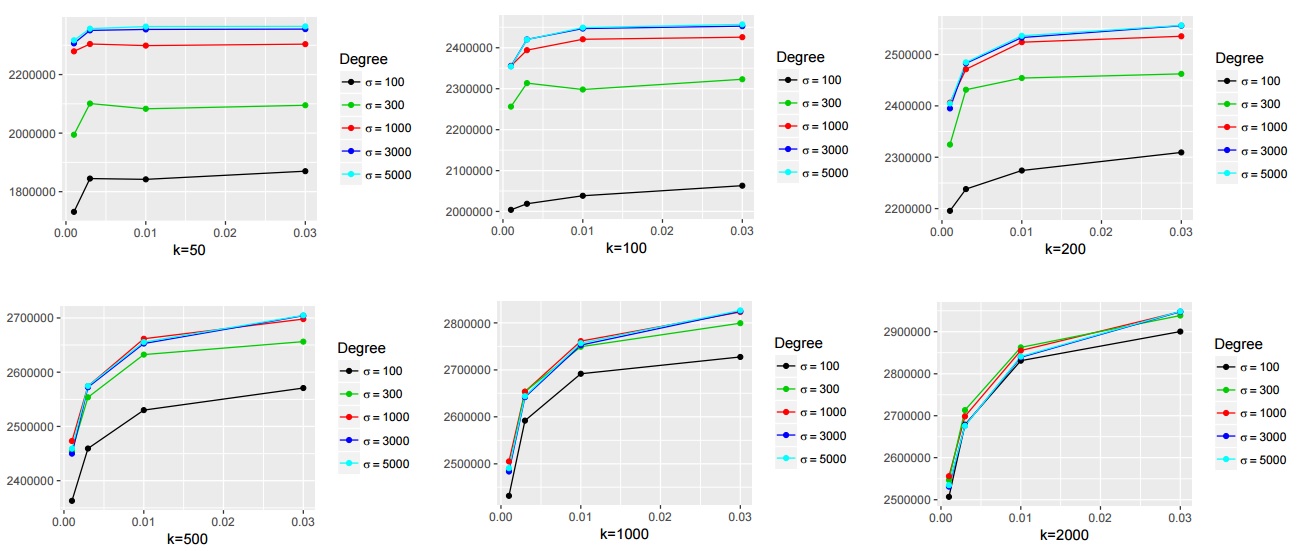}
	\caption{For the \problem{dominating-set} instance \LJinst,
          these plots show the number of covered nodes against the
          relative size of the sketch with
          $\rho\in[10^{-3},3\cdot10^{-2}]$,
%$\rho\in\{10^{-3},3\cdot10^{-3},10^{-2},3\cdot10^{-2}\}$, 
          $\sigma\in[100, 5000]$, and $k\in [10^2,10^4]$.  Curves in
          one plot correspond to different choices for $\sigma$.
          With large $\sigma$, the results of some runs are
          indistinguishable from the one next to it in the plot, hence
          invisible.}
	\label{fig:lj}
\end{figure*}

\begin{figure}
  \centering \includegraphics[scale=0.55]{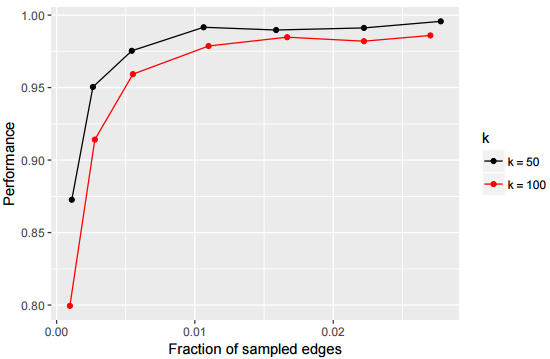}  %0.4
  \caption{The results for \DBLPinst\ are shown for
    $\rho\in[2\cdot10^{-3},5\cdot10^{-2}]$, $\sigma=100$.  We plot our
    performance relative to \stochgreedy\ against the fraction of
    edges from the input graph retained in our sketch.}
  \label{fig:approx-s-size} \end{figure}

\begin{figure}[t]
  \centering
  \includegraphics[scale=0.55]{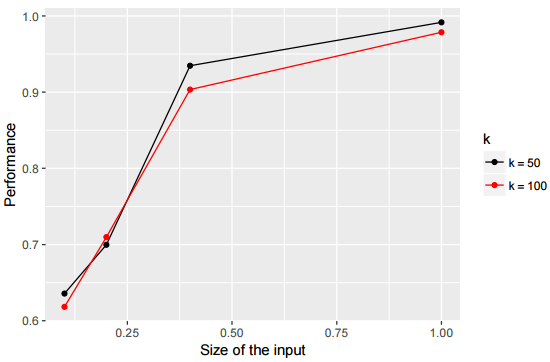}  %52
  \caption{The above are results of running the algorithm on the
    sampled version of \DBLPinst\ with $\rho=0.02$, $\sigma=100$. The
    $x$ axis  denotes the size of the sampled graph relative to
    the whole. The $y$ axis shows the quality relative to
    \stochgreedy.}
  \label{fig:approx-inp-size}
\end{figure}

\paragraph{DBLP coauthorship network} 
Figure~\ref{fig:approx-s-size} compares results of our
algorithm on \DBLPinst\ (with a range of parameters) to 
that of \stochgreedy. Each point in these
plots represents the mean of three independent runs. Interestingly, a
sketch with merely $3\%$ of the memory footprint of the input graph
attains $\%99.6$ of the quality of \stochgreedy.

%The first dataset that we consider is based on DBLP co-authorship network. DBLP co-authorship network contains $~3\times 10^5$ vertices and $~10^6$ edges. In the coverage instance that we construct each set or each element corresponds exactly to one vertex  from DBLP co-authorship network. Each set covers the elements correspond to the vertices within distance $3$ of it. The coverage instance that we construct based on this contains $~6\times 10^5$ vertices and $~2.6\times 10^7$ edges.

\begin{figure}[h]
  \centering
  \includegraphics[scale=0.61]{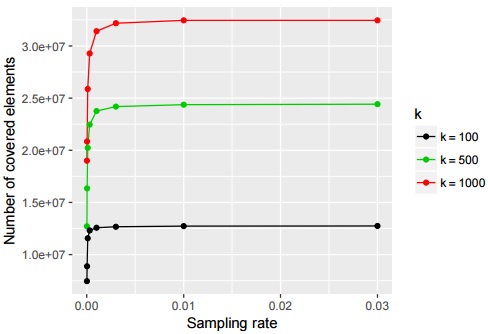}
  \caption{%The above %shows the outcome of our algorithm on
%plots 
Here we plot
the number of covered bigrams against $\rho$ for
\Guteinst\ with $\rho\in[10^{-5},3\cdot10^{-2}]$, $\sigma\in
                [10^2,10^4]$, and $k\in [10^2,1000^3]$.
%    We plot the number of covered bigrams against $\rho$. 
The curves
    corresponding to different values of $\sigma$ are practically
    indistinguishable.}
  \label{fig:gut}
\end{figure}

We run our algorithm on induced subgraphs of \DBLPinst\ of varying
sizes; %, and fix parameters $\rho=0.02$, $\sigma=100$ and $k\in
%\{50,100\}$; 
see Figure~\ref{fig:approx-s-size}. % shows the the results.
%Once again, each point represents the mean of three independent runs.
Interestingly, the performance of our algorithm improves the larger
the sampled graph becomes. In other words, if one finds parameters
$\rho$ and $\sigma$ on a subgraph of the input and applies it to the
whole graph, one does not lose much in the performance.  %% why?

%Figure~\ref{fig:deg} compares the results of our algorithm for
%$\sigma=100$ and no degree bound ($\sigma=\infty$).  We observe that
%picking $\sigma=100$ does not harm  the performance significantly.

%

\paragraph{Project Gutenberg dataset }
We run our algorithm on \Guteinst\ with different values for
$\rho$ and $\sigma$.  As shown in Figure~\ref{fig:gut}, the outcome of
the algorithm converges quickly.  In other words, for $\rho= 0.003$ and
$\sigma = 100$, the outcome of \stochgreedy\ on our sketch and on the
input graph are quite similar, while our sketch is $600$ times
smaller.

\paragraph{Other datasets}

%\xxx{this should probably move to the end of experimental section.}
%We now present some results for the datasets we did not get to report
%on in details. 
Due to space constraints we cannot report detailed results for the
other datasets.  However, Table~\ref{tab:misc} shows that for these
datasets, a small sketch suffices to get close to the single-machine
greedy solution.  In fact, these are small enough for the greedy
algorithm to run on one machine.

\begin{table}[h!]
\caption{Results for other datasets.}\label{tab:misc}
\vskip2mm
\centering
\begin{tabular}{l<{\hskip-3mm}r<{\hskip-3mm}r|l<{\hskip-4mm}r<{\hskip-3mm}c}
\hline
Instance    & Sketch & Quality &%\\
Instance    & Sketch & Quality \\
%Instance    & \multicolumn{1}{c}{Footprint} & \multicolumn{1}{c}{Quality} &%\\
%Instance    & \multicolumn{1}{c}{Footprint} & \multicolumn{1}{c}{Quality} \\
\hline
\WikiMain   & 0.06\% & 94.4\% &%\\
\DBLPinstx  & 1.7\% & 92\% \\
\WikiMain   & 2.4\% & 99.5\% &%\\
\DBLPinstx  & 3.1\% & 96\% \\
\WikiMain   & 7.7\% & 99.9\% &%\\
\Reuters    & 1.2\% & 87\% \\
\WikiTalk   & 1.5\% & 99.2\% &%\\
\Reuters    & 3.6\% & 92\% \\
\Hardinst{A}& 8.2\% & 96\% &%\\
\Reuters    & 10\% & 96\% \\
\hline
%\end{tabular}\hspace{7mm}
%\begin{tabular}{lrr}
%\hline
%
%\hline
%
%\DBLPinstx  & 0.3\% & 70\% \\
%\\
%\hline
\end{tabular}
\end{table}

The \LJinstx\ instance is too big for the single-machine greedy
algorithm.  Still we can compare our result to what is achievable for
a 10\% sample of the instance (with about 340 million edges).  With a
0.8\% footprint we obtain a solution of essentially the same quality.
With footprints 0.3\%, 0.2\%, 0.1\% and 0.75\%, we lose no more than
1\%, 3\%, 4\% and 9\%, respectively.

Except for the smallest, the planted instances are also too big for the
greedy algorithm.  Nonetheless, looking at the numbers, e.g.,
for \Hardinst{B}, we notice that the quality of the greedy solution
is almost the same for sketches of relative
sizes 0.3\% and 42\%---the latter has about 500 million edges.  In
particular, sketches of relative size 0.3\%, 1\% and 10\% produce
3\%, 2\% and 1\% error, respectively, compared to the sketch of size
42\%.  The results are similar for the other two planted instances.

\subsection{Feature-selection Problem}
\label{sec:featureselection}
Our algorithm is applicable to the \problem{feature-selection}
problem, which is a first step in many learning-based
applications~\cite{GE2003:feature}.
It is often too expensive to carry out a learning task on the entire
matrix or there might be overfitting concerns.  Typically a small
subset of ``representative'' features are picked carefully, so as not
to affect the overall quality of the learning task.  In practice, we
gauge the performance of feature selection by {\em reconstruction
  error} or {\em prediction accuracy}; see~\cite{ABFMRZ16:ICML} for
details of evaluation criteria.

In order to compare our preliminary results to previous
work~\cite{ABFMRZ16:ICML}, we model the problem as a \problem{maximum
  $k$-cover} instance by treating columns (i.e., features) as sets and
{\em pairs of rows} (i.e., pairs of sample points) as elements. We say a row {\em covers} a pair
of rows, if that column (feature) is active for both of those rows (sample points), and seek to pick $k$ columns that
{\em cover as many pairs the rows} as possible~\footnote{We also studied covering rows as opposed to covering pairs of rows, but that approach was not effective.}.

%In order to
%avoid noninformative columns that appear in all rows, we remove
%high-degree columns.  Results are slightly improved when we define the
%coverage as the number of pairs of rows that appear together in a
%selected column.

Table~\ref{tab:feature} compares our results to prior work.  Numbers
show prediction accuracy in percentage.  For description of the data
set and the first four algorithms, see~\cite{ABFMRZ16:ICML}.  We note
that these algorithms may only run on a $8\%$ sample of the dataset,
hence poorer performance compared to the latter two.
The fifth column exhibits a distributed version of \algoname{2-P} (the
two-phase optimization): Features are carefully partitioned across
many machines via taking into account some cut-based objective, and
then the two-phase optimization handles each part separately.  It is
noteworthy that the (distributed) partitioning phase itself takes
significant amount of time to run.
The last column corresponds to our distributed \problem{$k$-cover}
algorithm, which is more efficient than the algorithm of the fifth
column.  The results are similar to that of \algoname{Part}.

\begin{table}[ht]
  \caption{Results for \problem{feature selection} on
    \NewsTwenty\ dataset.}\label{tab:feature}\vspace{2mm}
  \begin{center}
    \begin{tabular}{rcccccc}
      \hline
      %\hspace{-1mm}
      \multicolumn{1}{c}{$k$} &
      {\algoname{Rnd}} & 
      {\algoname{\iffalse $2$-Phase\else$2$-P\fi}} & 
      {\algoname{\iffalse DistGreedy\else DG\fi}} & 
      {\algoname{PCA}} & 
      \hspace{-1mm}\algoname{Part}\hspace{-1mm} &  %% Clust?
      \algoname{Cover}\hspace{-1mm} 
      \\
      \hline
      500 & 54.9 & 81.8 &  80.2 & 85.8 & 84.5 & 86.2 \\
      1000 & 59.2 & 84.4 & 82.9 & 88.6 & 88.4 & 89.4 \\
      2500 & 67.6 & 87.9 & 85.5 & 90.6 & 92.3 & 91.2 \\
      \hline
    \end{tabular}
  \end{center}
\end{table}

%Our improvements are in the several areas: (1) The new algorithm
%inherently produces better results because it optimizes the whole
%dataset as one piece, rather than randomly breaking it up into smaller
%instances and solving each independently.  (2) Our algorithm is more
%efficient, hence we can run it on the entire dataset of about 1.5
%million features, whereas the previous results worked on an instance
%of about 100 hundred features (sampled uniformly).  (3) As our
%algorithm is pretty efficient and \NewsTwenty\ is a relatively
%small dataset, we can build an instance where elements to cover
%correspond to pairs of rows of the original matrix.  Therefore, our
%optimization takes into account a more subtle relationship between
%features, hence slightly better results.

We emphasize that we can run our algorithm on much larger datasets;
the evidence of this was provided above where we reported results for
\LJinst, for instance.

\section{Conclusions}
In this paper, we present almost optimal distributed algorithms for coverage problems. 
Our algorithms beat the previous ones in several fronts: e.g.,
(i) they provably acheive the optimal approximation factors for these problems, 
(ii) they run in only four rounds of computation (as opposed to logarithmic number of rounds), and
(iii) their space complexity is independent of the number of elements in the ground set.
Moreover, our algorithms can handle coverage problems with huge subsets 
(in which even one subset of the input may not fit on a single machine). 
Our empirical study shows practical superiority of our algorithms. 
Finally, we identified a new application of our algorithms in feature selection,
and presented preliminary results for this application. It would be nice to explore
this application in more details in the future.
%This technique also applies to streaming/RAM models, and it would be nice to  perform a formal empirical study of these algorithms in those models in the future.

{%\small
\bibliographystyle{plain}
\bibliography{k-cover}

\begin{thebibliography}{33}
\providecommand{\natexlab}[1]{#1}
\providecommand{\url}[1]{\texttt{#1}}
\expandafter\ifx\csname urlstyle\endcsname\relax
  \providecommand{\doi}[1]{doi: #1}\else
  \providecommand{\doi}{doi: \begingroup \urlstyle{rm}\Url}\fi

\bibitem[Abbassi et~al.(2013)Abbassi, Mirrokni, and Thakur]{AMT}
Abbassi, Zeinab, Mirrokni, Vahab~S., and Thakur, Mayur.
\newblock Diversity maximization under matroid constraints.
\newblock In \emph{KDD}, pp.\  32--40, 2013.

\bibitem[Aho \& Hopcroft(1974)Aho and Hopcroft]{aho1974design}
Aho, Alfred~V and Hopcroft, John~E.
\newblock \emph{Design \& Analysis of Computer Algorithms}.
\newblock Pearson Education India, 1974.

\bibitem[Altschuler et~al.(2016)Altschuler, Bhaskara, Fu, Mirrokni,
  Rostamizadeh, and Zadimoghaddam]{ABFMRZ16:ICML}
Altschuler, Jason, Bhaskara, Aditya, Fu, Gang, Mirrokni, Vahab~S.,
  Rostamizadeh, Afshin, and Zadimoghaddam, Morteza.
\newblock Greedy column subset selection: New bounds and distributed
  algorithms.
\newblock In \emph{ICML}, pp.\  2539--2548, 2016.

\bibitem[Anonymous()]{ours}
Anonymous.
\newblock Almost optimal streaming algorithms for coverage problems.
\newblock Attached as a supplemental material, 2017.

\bibitem[Backstrom et~al.(2006)Backstrom, Huttenlocher, Kleinberg, and
  Lan]{dblp-snap}
Backstrom, L., Huttenlocher, D., Kleinberg, J., and Lan, X.
\newblock Group formation in large social networks: Membership, growth, and
  evolution.
\newblock In \emph{KDD}, 2006.

\bibitem[Badanidiyuru \& Vondr{\'a}k(2014)Badanidiyuru and
  Vondr{\'a}k]{badanidiyuru2014fast}
Badanidiyuru, Ashwinkumar and Vondr{\'a}k, Jan.
\newblock Fast algorithms for maximizing submodular functions.
\newblock In \emph{SODA}, pp.\  1497--1514, 2014.

\bibitem[Badanidiyuru et~al.(2012)Badanidiyuru, Dobzinski, Fu, Kleinberg,
  Nisan, and Roughgarden]{badanidiyuru2012sketching}
Badanidiyuru, Ashwinkumar, Dobzinski, Shahar, Fu, Hu, Kleinberg, Robert, Nisan,
  Noam, and Roughgarden, Tim.
\newblock Sketching valuation functions.
\newblock In \emph{SODA}, pp.\  1025--1035, 2012.

\bibitem[Badanidiyuru et~al.(2014)Badanidiyuru, Mirzasoleiman, Karbasi, and
  Krause]{karbasiKDD2014}
Badanidiyuru, Ashwinkumar, Mirzasoleiman, Baharan, Karbasi, Amin, and Krause,
  Andreas.
\newblock Streaming submodular maximization: Massive data summarization on the
  fly.
\newblock In \emph{KDD}, 2014.

\bibitem[Bateni et~al.(2014)Bateni, Bhashkara, Lattanzi, and Mirrokni]{BBLM14}
Bateni, MohammadHossein, Bhashkara, Aditya, Lattanzi, Silvio, and Mirrokni,
  Vahab.
\newblock Mapping core-sets for balanced clustering.
\newblock In \emph{NIPS}, 2014.

\bibitem[Bird(2006)]{NLTK}
Bird, Steven.
\newblock {NLTK}: The natural language toolkit.
\newblock In \emph{COLING-ACL}, pp.\  69--72, 2006.

\bibitem[Blelloch et~al.(2012)Blelloch, Simhadri, and Tangwongsan]{BST12}
Blelloch, Guy~E., Simhadri, Harsha~Vardhan, and Tangwongsan, Kanat.
\newblock Parallel and {I/O} efficient set covering algorithms.
\newblock In \emph{SPAA}, pp.\  82--90, 2012.

\bibitem[Borodin et~al.(2012)Borodin, Lee, and Ye]{pods12}
Borodin, Allan, Lee, Hyun~Chul, and Ye, Yuli.
\newblock Max-sum diversification, monotone submodular functions and dynamic
  updates.
\newblock In \emph{PODS}, pp.\  155--166, 2012.

\bibitem[Chang et~al.(2006)Chang, Dean, Ghemawat, Hsieh, Wallach, Burrows,
  Chandra, Fikes, and Gruber]{bigtable-paper}
Chang, Fay, Dean, Jeffrey, Ghemawat, Sanjay, Hsieh, Wilson~C., Wallach,
  Deborah~A., Burrows, Mike, Chandra, Tushar, Fikes, Andrew, and Gruber,
  Robert~E.
\newblock Bigtable: A distributed storage system for structured data.
\newblock In \emph{OSDI'06: Seventh Symposium on Operating System Design and
  Implementation}, 2006.

\bibitem[Chierichetti et~al.(2010)Chierichetti, Kumar, and Tomkins]{CKT10}
Chierichetti, Flavio, Kumar, Ravi, and Tomkins, Andrew.
\newblock {Max-Cover} in {Map-Reduce}.
\newblock In \emph{WWW}, pp.\  231--240, 2010.

\bibitem[da~Ponte~Barbosa et~al.(2015{\natexlab{a}})da~Ponte~Barbosa, Ene,
  Nguyen, and Ward]{BENW15}
da~Ponte~Barbosa, Rafael, Ene, Alina, Nguyen, Huy~L., and Ward, Justin.
\newblock The power of randomization: Distributed submodular maximization on
  massive datasets.
\newblock In \emph{ICML}, pp.\  1236--1244, 2015{\natexlab{a}}.

\bibitem[da~Ponte~Barbosa et~al.(2015{\natexlab{b}})da~Ponte~Barbosa, Ene,
  Nguyen, and Ward]{BENW15b}
da~Ponte~Barbosa, Rafael, Ene, Alina, Nguyen, Huy~L., and Ward, Justin.
\newblock A new framework for distributed submodular maximization.
\newblock \emph{CoRR}, abs/1507.03719, 2015{\natexlab{b}}.

\bibitem[Dean \& Ghemawat(2004)Dean and Ghemawat]{osdi-DG04}
Dean, Jeffrey and Ghemawat, Sanjay.
\newblock {MapReduce}: Simplified data processing on large clusters.
\newblock In \emph{OSDI}, pp.\  137--150, 2004.

\bibitem[Gutenberg()]{gute:eng}
Gutenberg.
\newblock Search {P}roject {G}utenberg.
\newblock \url{https://www.gutenberg.org/ebooks/}, 2016.
\newblock Accessed: 05-19-2016.

\bibitem[Guyon \& Elisseeff(2003)Guyon and Elisseeff]{GE2003:feature}
Guyon, Isabelle and Elisseeff, Andr{\'e}.
\newblock An introduction to variable and feature selection.
\newblock \emph{J. Mach. Learn. Res.}, 3:\penalty0 1157--1182, 2003.

\bibitem[Indyk et~al.(2014)Indyk, Mahabadi, Mahdian, and Mirrokni]{IMMM14}
Indyk, Piotr, Mahabadi, Sepideh, Mahdian, Mohammad, and Mirrokni, Vahab.
\newblock Composable core-sets for diversity and coverage maximization.
\newblock In \emph{PODS}, 2014.

\bibitem[{Infochimps, Inc.}(2016)]{enwords.txt}
{Infochimps, Inc.}
\newblock English word list.
\newblock \url{https://github.com/dwyl/english-words/blob/master/words.txt},
  2016.

\bibitem[Karloff et~al.(2010)Karloff, Suri, and Vassilvitskii]{soda-KSV10}
Karloff, Howard~J., Suri, Siddharth, and Vassilvitskii, Sergei.
\newblock A model of computation for {MapReduce}.
\newblock In \emph{SODA}, pp.\  938--948, 2010.

\bibitem[Kumar et~al.(2013)Kumar, Moseley, Vassilvitskii, and Vattani]{KMVV13}
Kumar, Ravi, Moseley, Benjamin, Vassilvitskii, Sergei, and Vattani, Andrea.
\newblock Fast greedy algorithms in {MapReduce} and streaming.
\newblock In \emph{SPAA}, pp.\  1--10, 2013.

\bibitem[Leskovec et~al.(2010)Leskovec, Huttenlocher, and Kleinberg]{wiki-snap}
Leskovec, J., Huttenlocher, D., and Kleinberg, J.
\newblock Governance in social media: A case study of the {W}ikipedia promotion
  process.
\newblock In \emph{AAAI}, 2010.

\bibitem[Lewis et~al.(2004)Lewis, Yang, Rose, , and Li]{reuters}
Lewis, D.~D., Yang, Y., Rose, T., , and Li, F.
\newblock {RCV1}: A new benchmark collection for text categorization research.
\newblock \emph{Journal of Machine Learning Research}, 5:\penalty0 361--397,
  2004.

\bibitem[Mirrokni \& Zadimoghaddam(2015)Mirrokni and Zadimoghaddam]{MZ15}
Mirrokni, Vahab~S. and Zadimoghaddam, Morteza.
\newblock Randomized composable core-sets for distributed submodular
  maximization.
\newblock In \emph{STOC}, pp.\  153--162, 2015.

\bibitem[Mirzasoleiman et~al.(2013)Mirzasoleiman, Karbasi, Sarkar, and
  Krause]{nips13}
Mirzasoleiman, Baharan, Karbasi, Amin, Sarkar, Rik, and Krause, Andreas.
\newblock Distributed submodular maximization: Identifying representative
  elements in massive data.
\newblock In \emph{NIPS}, pp.\  2049--2057, 2013.

\bibitem[Mirzasoleiman et~al.(2014)Mirzasoleiman, Badanidiyuru, Karbasi,
  Vondr{\'a}k, and Krause]{mirzasoleiman2014lazier}
Mirzasoleiman, Baharan, Badanidiyuru, Ashwinkumar, Karbasi, Amin, Vondr{\'a}k,
  Jan, and Krause, Andreas.
\newblock Lazier than lazy greedy.
\newblock \emph{arXiv preprint arXiv:1409.7938}, 2014.

\bibitem[Mirzasoleiman et~al.(2015)Mirzasoleiman, Karbasi, Badanidiyuru, and
  Krause]{nips15}
Mirzasoleiman, Baharan, Karbasi, Amin, Badanidiyuru, Ashwinkumar, and Krause,
  Andreas.
\newblock Distributed submodular cover: Succinctly summarizing massive data.
\newblock In \emph{NIPS}, 2015.

\bibitem[Mirzasoleiman et~al.(2016)Mirzasoleiman, Zadimoghaddam, and
  Karbasi]{MZK16}
Mirzasoleiman, Baharan, Zadimoghaddam, Morteza, and Karbasi, Amin.
\newblock Fast distributed submodular cover: Public-private data summarization.
\newblock In \emph{NIPS}, pp.\  3594--3602, 2016.

\bibitem[Saha \& Getoor(2009)Saha and Getoor]{saha2009maximum}
Saha, Barna and Getoor, Lise.
\newblock On maximum coverage in the streaming model \& application to
  multi-topic blog-watch.
\newblock In \emph{SDM}, volume~9, pp.\  697--708, 2009.

\bibitem[Wolff(2016)]{python:gute}
Wolff, Clemens.
\newblock Python package {G}utenberg.
\newblock \url{https://pypi.python.org/pypi/Gutenberg}, 2016.
\newblock Version 0.4.2.

\bibitem[Yang \& Leskovec(2012)Yang and Leskovec]{lj-snap}
Yang, J. and Leskovec, J.
\newblock Defining and evaluating network communities based on ground-truth.
\newblock In \emph{ICDM}, 2012.

\end{thebibliography}
}

%\section{Sketch-based algorithms}\label{sec:algorithms}
%\input{algorithms.tex}

\appendix
\section{Omitted Proofs}
We use the following lemma to prove Theorems~\ref{thm:mr:kcover} and~\ref{thm:mr:epscover}.

\begin{lemma} [From~\cite{ours}]\label{thm:sketchAlg}
	For any $\eps\in (0,1]$ and any set-cover graph $G$, there exist
	sketch-based algorithms that succeed with probability
	$1-\frac 1 n$ in finding the following.
	\vspace{-2ex}
	\begin{packed_enum}
		\item
		One finds a $(1-\frac 1 e -\eps)$-approximate solution to
		\problem{$k$-cover} on $G$, working on a sketch with
		$\tilde{O}(n)$ edges.
		\item
		The other  %%, given $C\geq 1$, 
		finds a $(1+\eps)\log\frac 1
		{\lambda}$ approximate solution to \problem{set cover with $\lambda$
			outliers} on $G$. The sketches used altogether have
		$\tilde{O}(n/\lambda^3) =\tilde{O}_{\lambda}(n)$ edges.
	\end{packed_enum}
	\vspace{-2ex}
\end{lemma}

%% \begin{proofof}{Theorem~\ref{thm:mr:kcover}} 
%%   We use Algorithm~\ref{Alg:MR} to construct $H_{\leq n}$.  This
%%   succeeds with probability $1-\frac 1 {n^2}$ by
%%   Theorem~\ref{thm:Hn4}.  Then we apply Lemma~\ref{thm:sketchAlg}
%%   to solve the problem in the last round.
%% %By
%% %  Theorem~\ref{thm:Hn4}, with probability $1-\frac 1 {n^2}$ the subgraph received by
%% %  machine one in the fourth round is $H_{\leq n}$. 
%%   This guarantees  a $(1-\frac 1 e
%%   -\epsilon)$-approximate solution to \problem{$k$-cover}, with probability
%%   $1-\frac 1 n - \frac 1 {n^2} > 1 - \frac 2 n$.
%% \end{proofof}

\begin{proofof}{Theorem \ref{thm:mr:epscover}}
	We run $\log_{1+\eps/3} n$ copies of the first three stages of
	Algorithm~\ref{Alg:MR} (simultaneously) to construct the
	$\log_{1+\eps/3} n$ different sketches required by
	Lemma~\ref{thm:sketchAlg}. %Theorem~\ref{thm:Hn4} 
        The discussion in Theorem~\ref{thm:mr:kcover} implies that each
	copy of the sketch is constructed 
	correctly, with probability $1-\frac 1 {n^2}$. Together
	with Lemma~\ref{thm:sketchAlg} this proves that our algorithm
	gives a $(1+\eps)\log\frac 1 {\lambda}$-approximate solution
	to \problem{set cover with $\lambda$ outliers}, with
	probability $1-\frac 1 n - \log_{1+\eps/3} n\frac 1 {n^2} > 1
	- \frac 2 n$.
\end{proofof}

\section{Description of Datasets}
\label{sec:full-data}

%We first give an overview of what datasets we use in our empirical
%study.  Afterwards we mention the methodology for our experiments
%before discussing the individual results.
We run our empirical study on five types of instances. 
A summary is presented in Table~\ref{tab:datasets}. 
\problem{Domainating set} instances include \LJinst, \LJinstx, \DBLPinst\ and
\DBLPinstx,  where the objective is to cover the nodes via multi-hop
neighborhoods.
We have two sets of {\em bag-of-words} instances, where the goal is to cover
as many words/bigrams via selecting a few documents: \Guteinst\ and
\SGuteinst\ for books and \Reuters\ for news articles.
Instances \WikiMain\ and \WikiTalk\ are our contribution graphs where we want to
find a set of users who revised many articles.
Finally we have some \problem{planted set-cover} instances that are known
to fool the greedy algorithm: \Hardinst{A}, etc.

\paragraph{Dominating-set instances}
%We used two instances generated by \cite{AleNIPS}\footnote{Simultaneous submission to NIPS, hence author names anonymized for now.}.  
We build \problem{set-cover} instances based on
graphs for LiveJournal (social network)~\cite{lj-snap} and DBLP
(database of coauthorship)~\cite{dblp-snap}. % from SNAP~\cite{SNAP}.
The vertices of these graphs form both $\sets$ and $\elements$ in the
new instances \LJinst\ and \DBLPinst, and a set $S$ covers an element $v$
if and only if $v$ is within the three-hop neighborhood of $S$ in the
original graph---reachable by a path of length at most
three.  Similarly, we build two smaller instances \DBLPinstx\ and
\LJinstx\ that are based on
2-hop neighborhoods.  To show scalability further, we also use
a sampled version of \LJinst: each vertex is picked with a fixed
probability.

\paragraph{``Bag of words'' instances}
%% We could have pointed to the main homepage via {proj:gute}.
We build \Guteinst\ based on documents and bigrams
inside them.  The starting point is the set of 50,284 books on Project
Gutenberg with IDs less than 53,000~\cite{gute:eng},\footnote{The
	upper bound was picked because the project claims to host
	52,031 books~\cite{gute:eng}.}  
downloaded via% Python package Gutenberg~\cite{python:gute}.  
~\cite{python:gute}.
We then remove the English
stopwords~\cite{NLTK}, and throw away 8,568 books
we think are not in English, leaving us with 41,716 books.  This
was done heuristically: any book with more than half its distinct
words missing from an English word list~\cite{enwords.txt} was deemed
non-English. (Non-English books in the collection
make the \problem{set-cover} instances much simpler, since picking books from
different languages allows us to cover a lot of new words/bigrams.
This process reduces the number of distinct words by about 63\%.)
Natural Language Toolkit~\cite{NLTK} was then used to turn words into
their stems, before generating the list of bigrams in each book.
A smaller version of this dataset, called \SGuteinst, was also
generated using books with IDs less than 1,000.
We also consider another bag-of-words instance, \Reuters, which is a
collection of Reuters news articles~\cite{reuters} written 1996--1997.
The words in each article have already been changed to their stems.  %That is
%what we use.  
There are four medium-size subdatasets % and a very small one
in the collection.  We only report results on the \texttt{p0} part.
The others behave similarly.

\paragraph{Contribution graphs}
The Wikipedia edit history (until 2008) is available~\cite{wiki-snap}.
We build two datasets \WikiMain\ and
\WikiTalk\ based on this.  
Users have made revisions in either namespace (main article texts or talk pages).
We place edges between
users and pages they have revised.  A \problem{set-cover}
solution then consists of a group of users who have revised all (or
many) pages.

\paragraph{Planted set-cover instances}
We also generate instances where a small set cover is planted in an
otherwise random graph. %~\cite{AleNIPS}.  
The advantage is that we know the optimum solution even
for large instances.  We build such
graphs with parameters $k$, $m$, $k'$ and $\epsilon$.  Such instances
have $m$ elements and $k+k'$ sets,  $k$ of which
have the same size and perfectly cover the entire ground set;
the other $k'$ have random elements but are a factor
$1+\epsilon$ larger than the planted sets.  We use four such graphs in
our experiments with %the following parameters.
parameters mentioned in Table~\ref{tab:planted}.

%\centerline{
\begin{table}
	\caption{Parameters used to generated ``planted'' datasets.}\label{tab:planted}
	\vskip2mm
	\centering
	\begin{tabular}{lrrrr}
		\hline
		Name & \multicolumn{1}{c}{$k$} & \multicolumn{1}{c}{$m$} & \multicolumn{1}{c}{$k'$} & \multicolumn{1}{c}{$\epsilon$} \\
		\hline
		\Hardinst{A} & 100   & 10,000     & 10,000  & 0.2 \\
		\Hardinst{B} & 100   & 1,000,000  & 100,000 & 0.2 \\
		%\hline
		%\end{tabular}\hspace{2mm}
		%\begin{tabular}{lrrrr}
		%	\hline
		%Name & \multicolumn{1}{c}{$k$} & \multicolumn{1}{c}{$m$} & \multicolumn{1}{c}{$k'$} & %\multicolumn{1}{c}{$\epsilon$} \\
		%\hline	
		\Hardinst{C} & 500   & 10,000,000 & 100,000 & 0.2 \\
		\Hardinst{D} & 1,000 & 10,000,000 & 100,000 & 0.2 \\
		\hline
	\end{tabular}%}
\end{table}

\paragraph{Feature-selection instances}
Built from a column subset selection instance, sets correspond to
columns and elements corresopnd to rows.  We aim to pick a small
number of columns that collectively appear in a majority of rows.  We
focus on \NewsTwenty\ dataset that is discussed in detail
in~\cite{ABFMRZ16:ICML}.

\end{document}